\documentclass{article}

\usepackage{amsfonts,amsmath,amsthm}
\usepackage{graphicx}
\usepackage{lipsum}
\newtheorem{theorem}{Theorem}[section]

\newtheorem{lemma}[theorem]{Lemma}

\newcommand{\nub}{\boldsymbol\nu}
\newcommand{\E}{{\bf E}}
\newcommand{\grph}{G(\nub, \E)}

\newcommand{\spn}{\mathfrak{sp}(2n,\mathbb C)}
\newcommand{\spnone}{\mathfrak{sp}(2n+2,\mathbb C)}
\newcommand{\spfour}{\mathfrak{sp}(4,\mathbb C)}
\newcommand{\spsix}{\mathfrak{sp}(6,\mathbb C)}
\newcommand{\speight}{\mathfrak{sp}(8,\mathbb C)}
\newcommand{\ltc}[2]{\mathcal{L}(#1,#2)}
\newcommand{\ltct}[1]{\mathcal{L}_{#1}}
\newcommand{\hgrp}{G_h(\nub, \E)}
\newcommand{\mgrp}{G_m(\nub, \E)}
\newcommand{\sgrp}{G_s(\nub, \E)}
\newcommand{\lkpa}{\lambda_\kappa}
\newcommand{\lkpal}{\lambda_{\kappa'}}
\newcommand{\klam}{\kappa_\lambda}
\newcommand{\muzn}{[0,\mu)^n \cap \mathbb{Z}^n}
\newcommand{\dist}[1]{\delta_{#1}(\lkpa,\lkpal)}

\newcommand{\rwl}{{\em rwl} }

\begin{document}

\markboth{A. F. Ramos \& Y. Deng}
{Symmetry-guided design of topologies for supercomputer networks}


\title{Symmetry-guided design of topologies for supercomputer networks}

\author{Alexandre F. Ramos$^{1,2,3,4,5,*}$ and Yuefan Deng$^{5,6,7,**}$ \\
  1. Escola de Artes, Ci\^encias e Humanidades; \\
  2. N\'ucleo de Interdisciplinar em Sistemas Complexos; \\
  3. Dept. de Radiologia e Oncologia -- Fac. de Medicina; \\
  4. Instituto do C\^ancer do Estado de S\~ao Paulo, \\
  Universidade de S\~ao Paulo, S\~ao Paulo, SP, 03828-000, Brazil. \\
  5. Shandong Computer Science Center (National Supercomputer \\
  Centre in Jinan), Jinan, Shandong 250101, P.R. China. \\ 
  6. Department of Applied Mathematics and Statistics, Stony Brook \\
  University, Stony Brook, NY 11794, USA. \\
  7. Department of Data and Computer Science, Sun Yat-sen University, \\
  Guangdong 510006, P.R. China. \\
  $^*$alex.ramos@usp.br, $^{**}$ Yuefan.Deng@stonybrook.edu
}




\maketitle

\begin{abstract}
  
A family of graphs optimized as the topologies for supercomputer interconnection networks is proposed. The special needs of such network topologies, minimal diameter and mean path length, are met by special constructions of the weight vectors in a representation of the symplectic algebra. Such theoretical design of topologies can conveniently reconstruct the mesh and hypercubic graphs, widely used as today's network topologies. Our symplectic algebraic approach helps generate many classes of graphs suitable for network topologies.

MSC-class: 94C15, 68R10, 68M10, 22D10

ACM-class: C.2.1; G.2.0; G.2.1; G.2.2

\end{abstract}





\newpage

\section{Introduction}

Despite the processing speed of the fastest supercomputers are at the hundreds of petaflops levels \cite{top500}, applications keep demanding faster computations and put the exascale processing speeds as an essential goal for supercomputing \cite{Kogge11,Pawlowski10}. To reach such a goal we must design interconnect network topologies that enhance communication on parallel computers, among others \cite{Garzon2012}. Therefore, developing systematic tools for optimally interconnect the ever increasing number of communicating components of supercomputers is an important challenge for the HPC community \cite{Duato2002,Dally2003}. A requirement for those tools is to diminish the distance between software and hardware through a unified framework that would enable the development of architecture-aware programming tools for more efficient communication or task mapping \cite{Dongarra11a,Zhang14}.

In this manuscript we extend the symmetries approach for the design of supercomputer network topologies by deriving a general method to construct different network topologies from the Cartan classification of the Lie algebras \cite{Deng12}. The symmetries approach for the construction of interconnect topologies is based on the correspondence between a supercomputer, a graph, and the roots and weights lattice of an irreducible representation of a Lie algebra as resumed on the Table \ref{corrsp}. To construct interconnect topologies from the Lie symmetries we start considering the matrix representations of the Lie algebras on the Cartan basis. Then we linearly combine the ladder operators of the algebra so that the resulting matrix shall coincide with an adjacency matrix of a graph with certain properties. The labeling of the vertices of the graph is constructed in terms of the weight vectors of an irreducible representation of the algebra. As an example we apply our framework for the description of the mesh and hypercubic interconnect topologies in terms of the well known $\mathfrak{su}(2)$ algebra \cite{Duato2002,Efe1991,Harary1988}. It turns out that in this approach those two interconnect topologies are represented in the same framework. Then we introduce a new network topology based on the symplectic algebras that we call symplectic topologies. The graph properties of the symplectic topologies are evaluated in comparison with the hypercube and the mesh. We show that the symplectic algebras enable the construction of graphs characerized by shorter path length distances. As an example we construct a graph with symplectic topology composed by 4,882,813 vertices having an average path length of 8. As a comparison, we construct a hypercubic graph with 1,048,576 vertices and average path length given by 12. Therefore, the symplectic algebras may generate graphs whose interconnect topology is characterized by a high number of vertices having greater connectivity than usual topologies such as the hypercube.

\begin{table}[!h]
  \centering
  \begin{tabular}{c|c|c}
    \hline
    Computer Science & Graph theory & Group theory \\
    \hline
    \hline
    Supercomputer & Graph & Roots and Weights lattices \\
    Computing nodes; switch; router & Graph vertices & Weight vectors \\
    Network topology & Graph edges & Root vectors \\
    \hline
    \hline
  \end{tabular}
  \caption{This table gives a concise presentation of the use of group theoretical methods to design of supercomputer interconnect topology. We introduce the correspondence between terminology and concepts from the three different fields involved in our approach: computer science; graph theory; group theory. The generators of a Lie algebra can be represented in terms of vectors called roots. The generators of a Lie algebra act on a vector space spanned by vectors denoted as weights. The vector addition of two weights can be expressed as a linear combination of root vectors. Hence, the roots of an algebra are connecting the weights and correspond to the edges of a graph while the weights are corresponding to the vertices.}
\label{corrsp}
\end{table}

The remainder of this manuscript is divided as follows: we start with a brief review on the roots and weights lattices for the symplectic algebra. The next section is devoted to the presentation of our results on the construction of the symplectic graph from a roots and weights lattice; that is followed by a comparative analysis of network properties of the mesh, hypercube and symplectic graphs. Our conclusions are presented on the last section. 


\section{Roots and weights lattices of the symplectic algebra}

The symplectic graph is constructed in terms of the roots and weights lattices associated to the irreducible representations (irreps) of the symplectic algebra. The $\spn$ algebra has {\em rank} $n$ and a total of $n(2n+1)$ generators. The symplectic algebras are naturally written in the Cartan-Weyl basis, having a $n$ dimensional center (Cartan's subalgebra) and $2n^2$ {\em root generators} that are also called {\em ladder operators}.

The set $R$ of root vectors (here on referred to as roots) of the $\spn$ are $n$-dimensional vectors consisting of two sub-sets of roots, the long roots $R_l$ and the short roots $R_s$. Given a $\spn$ algebra, the set $R_l$ is composed by vectors of the form $(0,\dots,\pm 2,\dots,0)$ while vectors of the form $(0,\dots,\pm 1, \dots, \pm 1, \dots, 0)$ compose $R_s$. The cardinality of $R_l$ and $R_s$ is, respectively, $2n$ and $2n(n-1)$. A lexicographic evaluation of the roots of the $\spn$ permits classifying its roots as {\em negative roots} (or {\em positive roots}) according to the first non-null coordinate being negative (or positive).

\begin{figure}[!h]
\centering
\includegraphics[width=0.4\linewidth]{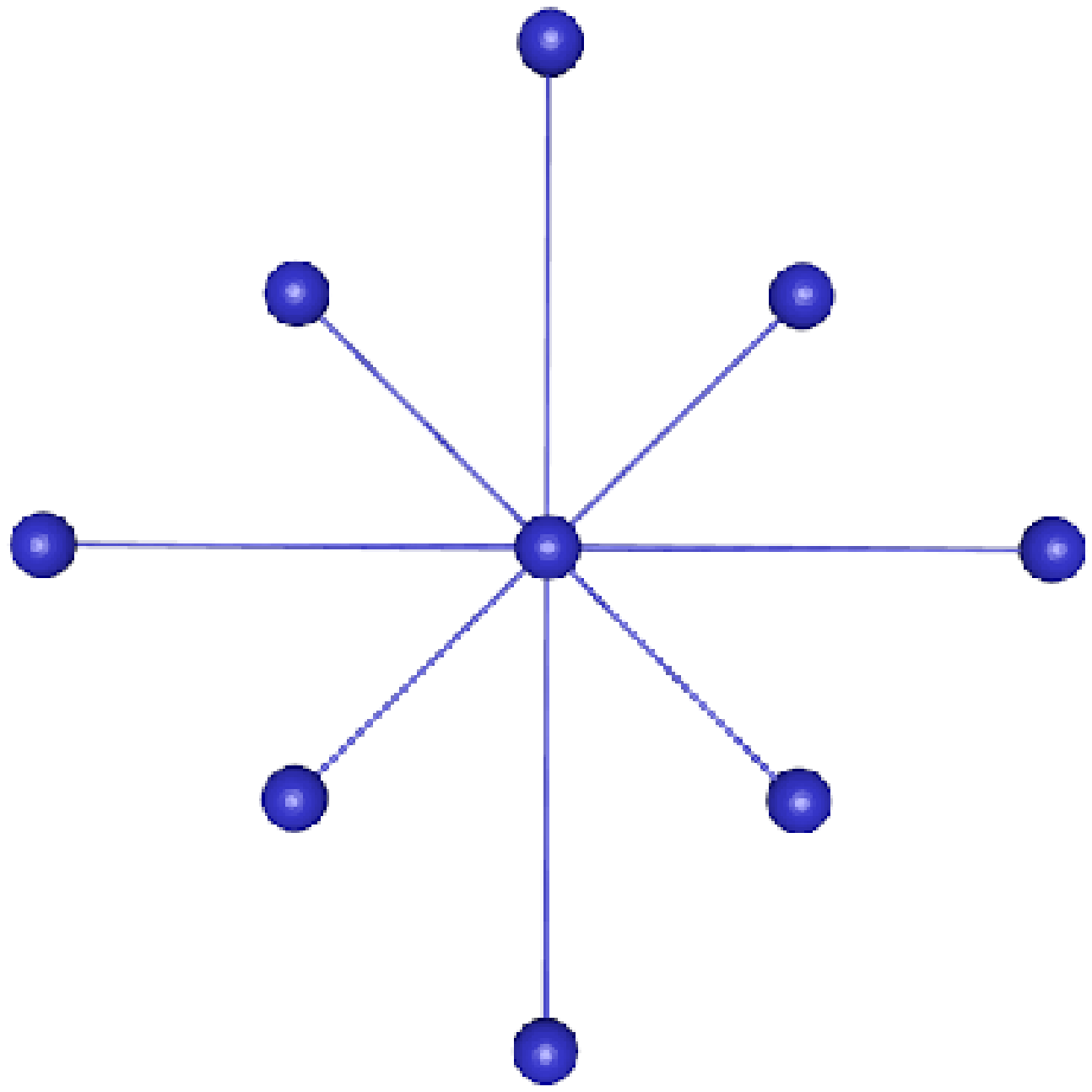}
\includegraphics[width=0.4\linewidth]{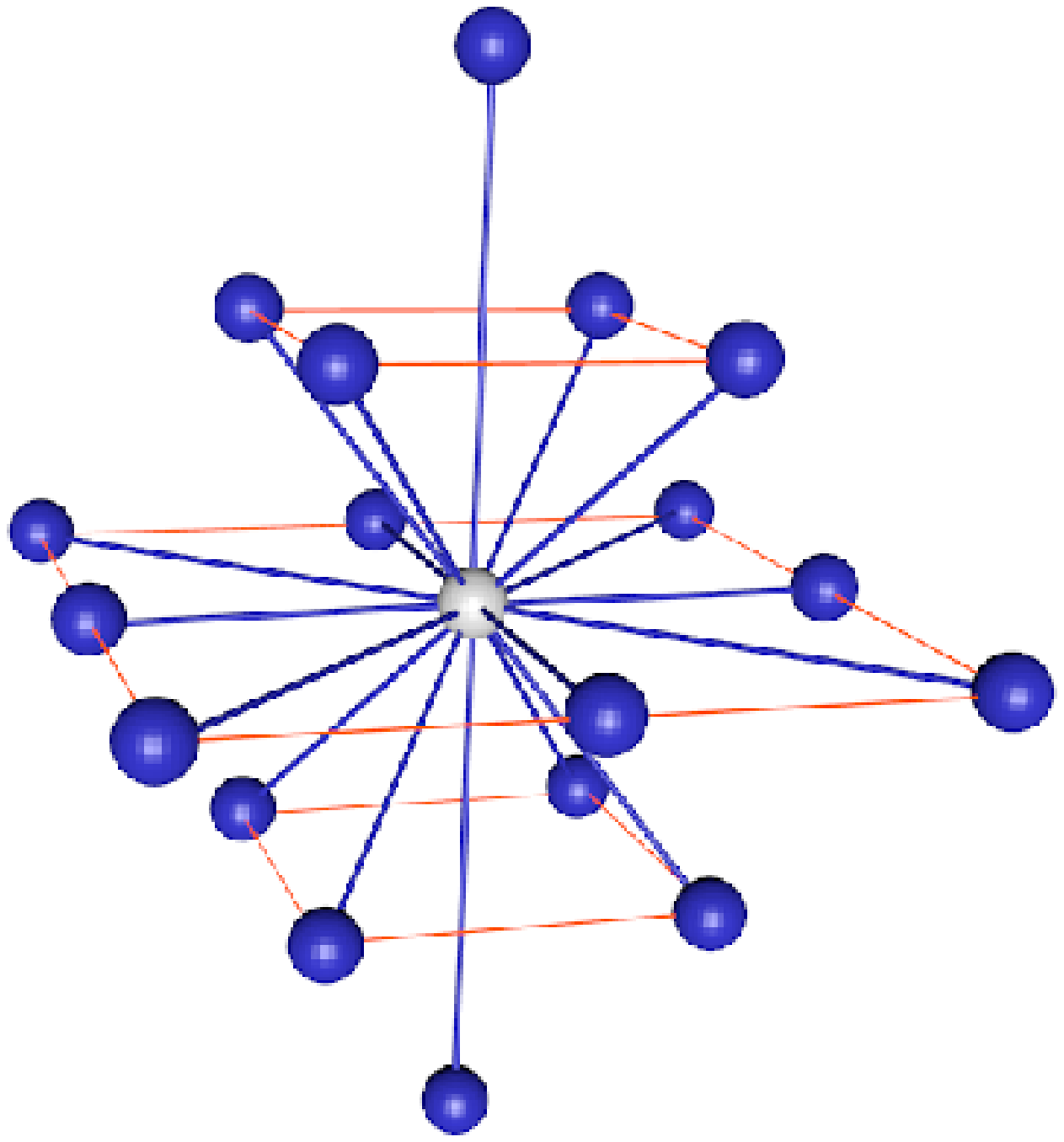}
\caption{The root diagram for the $\spfour$ and $\spsix$ algebras.}
\label{roots}
\end{figure}

The root diagrams of the $\spfour$ and of the $\spsix$ Lie algebras is indicated at the Figure \ref{roots}. Note that we may write the roots in terms of the unit vectors $e_i$ spanning the $n$-dimensional space of the root diagram. Namely,
\begin{equation}\label{roots1}
\pm 2 e_i \ \ {\rm and} \ \ \pm e_i \pm e_j, \ \ i,j = 1,\dots,n,
\end{equation}
respectively, indicate the long and the short roots of the symplectic algebra.

For a given irreducible representation (irrep) of $\spn$, a weight vector (here on referred to as {\em weight}) is a vector of the carrier space of the algebra that is an eigenvector of the Cartan subalgebra. A weight has coordinates indicated by an $n$-tuple $(m_1, \dots, m_n)$ and the highest weight of an irrep is labeled by $(M_1, \dots, M_n)$. The highest weight imposes a constraint on $m_i$ such that 
\begin{equation} \label{mvalues}
-M_i \le m_i \le M_i,
\end{equation}
with a relationship among them and the dimension of the roots and weights lattice of the $\spn$ given in terms of its highest weight by means of the Weyl formula \cite{Wybourne1974}.
 
\begin{figure}[!h]
  \centering
  \begin{minipage}[b]{0.42\linewidth}
    \includegraphics[width=\linewidth]{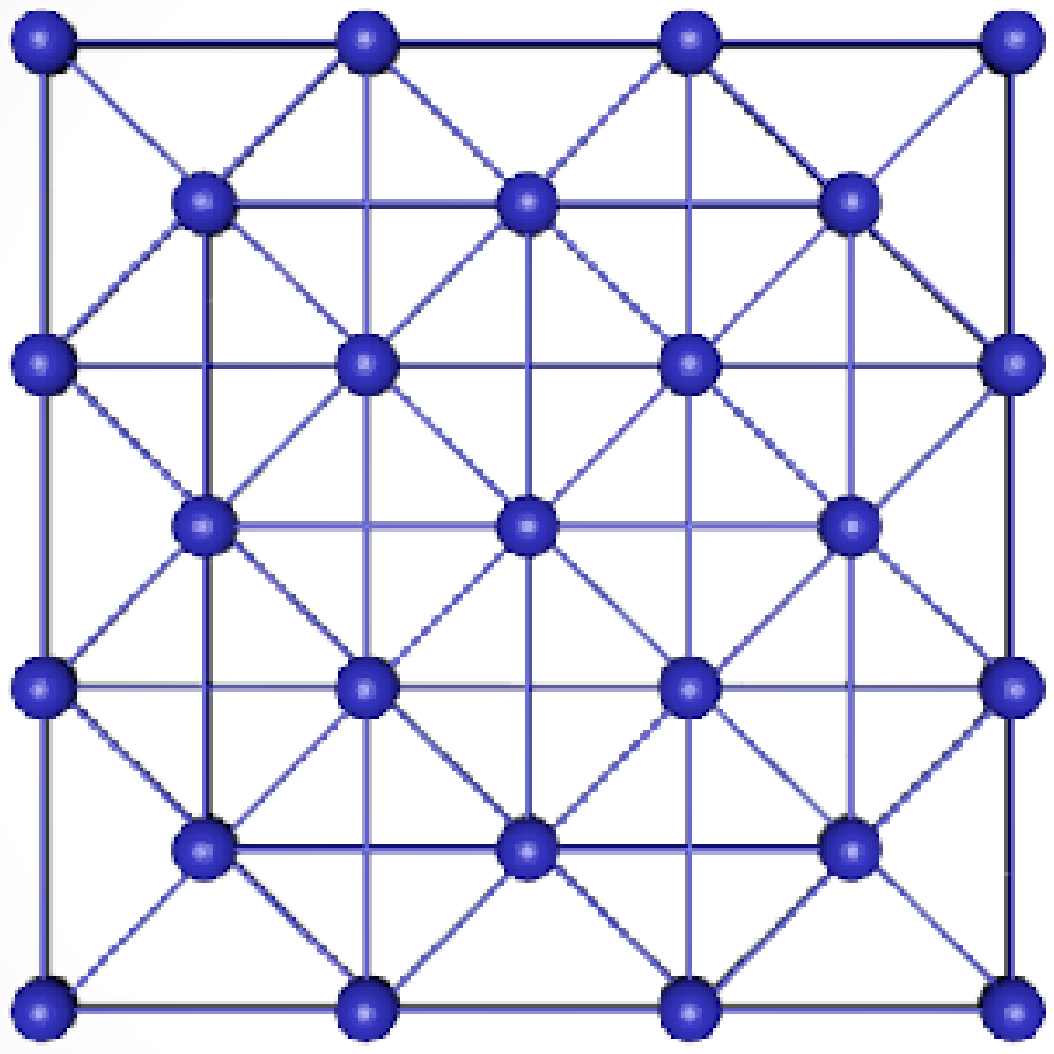}
    \end{minipage}
  \hfill
  \begin{minipage}[b]{0.42\linewidth}
    \includegraphics[width=\linewidth]{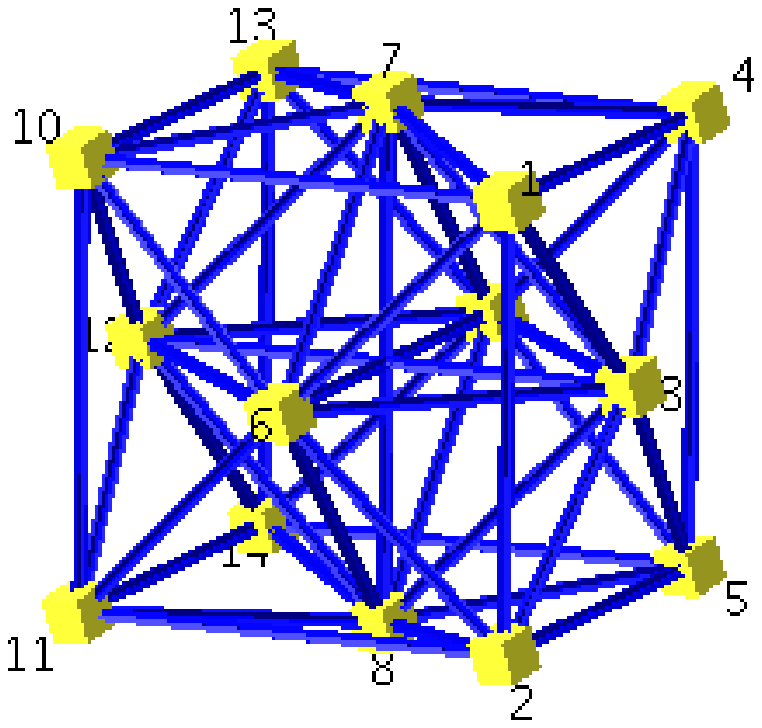}
  \end{minipage}
  \caption{The roots and weights lattices for the irrep $(3,3)$ and $(1,1,1)$, respectively, of the $\spfour$ and $\spsix$ is shown on the left and center frames. 
The lattices are both characterized by the existence of two types of mesh structures, the external and internal ones. The interconnects within a given mesh are given in terms of the long roots while the short roots are interconnecting the internal and external meshes.}
\label{smpltc}
\end{figure}

Interestingly, the fundamental representation of the $\spn$ will have the same number of nodes as the hypercubic topology. However, it shall generate a complete graph that is not of interest for our analysis. Here we focus on the graphs resulting from irreps whose highest weight have the form $(M,\dots,M)$, also known as the anti-symmetric irreps of the symplectic algebras. In that case, the symplectic graph has a mesh-like structure composed by multiple meshes. The external mesh shell encompasses a rich and intricate internal structure, composed by multiple mesh-like sub-graphs. The interconnect between weights of a given sub-mesh can be obtained in terms of the long roots of the $\spn$. Analogously, the interconnect between the weights of two different sub-meshes are given by means of the short roots of the $\spn$. The Figure \ref{smpltc} shows a sample of roots and weights lattices for anti-symmetric irreps of the $\spfour$ and $\spsix$ algebras. Due to the high vertex degree characterizing the roots and weights lattices of the irreps of the $\spsix$ we decided to only present a unit cell (center frame). 
The Figure \ref{graph} shows a graph representation for the roots and weights lattices of the irreps of the $\spfour$ and $\spsix$ algebras composed by 41 and 63 vertices.

\begin{figure}[!h]
  \centering
  \begin{minipage}[b]{0.45\linewidth}
    \includegraphics[width=\linewidth]{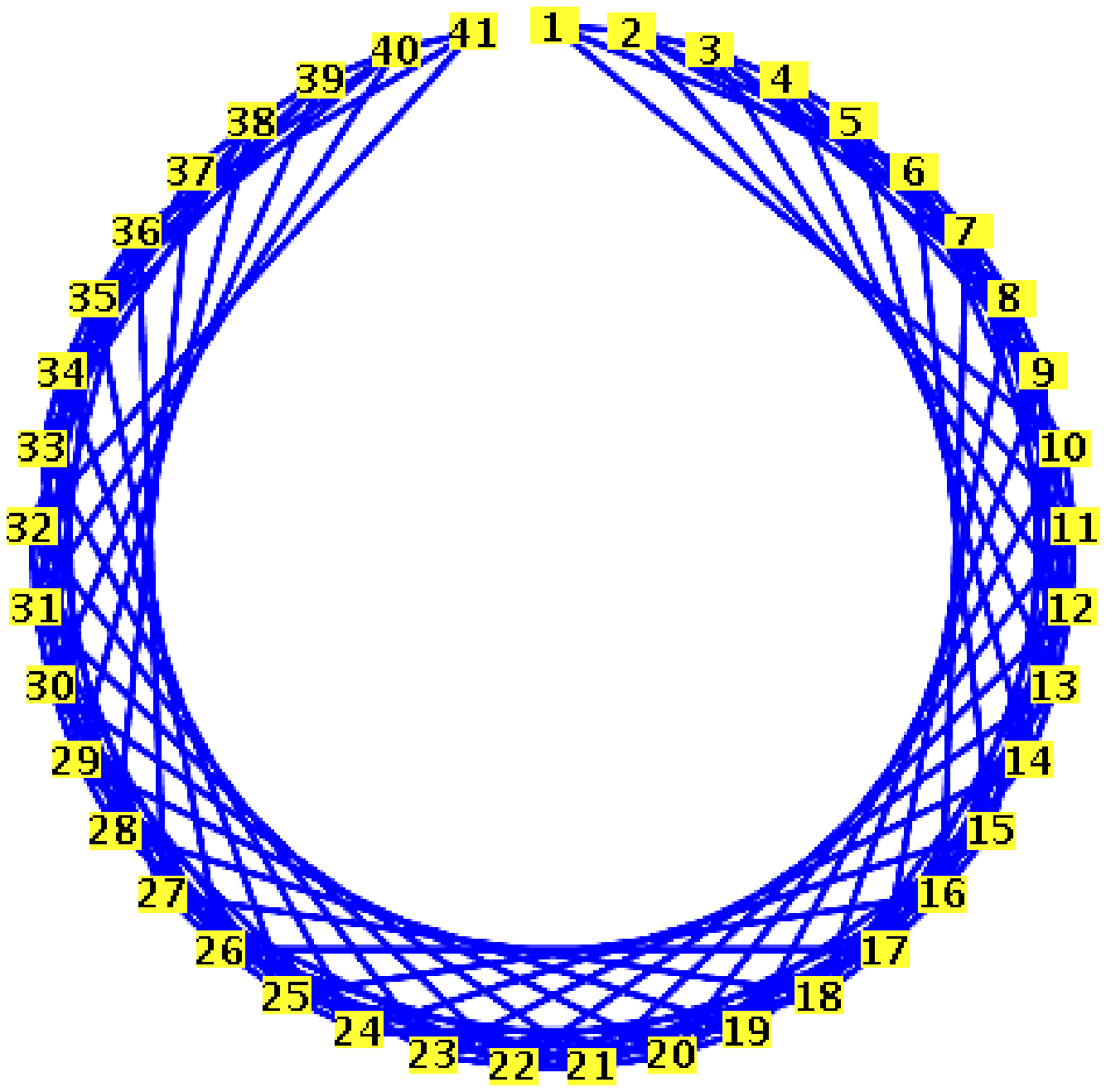}
    \end{minipage}
  \hfill
  \begin{minipage}[b]{0.45\linewidth}
    \includegraphics[width=\linewidth]{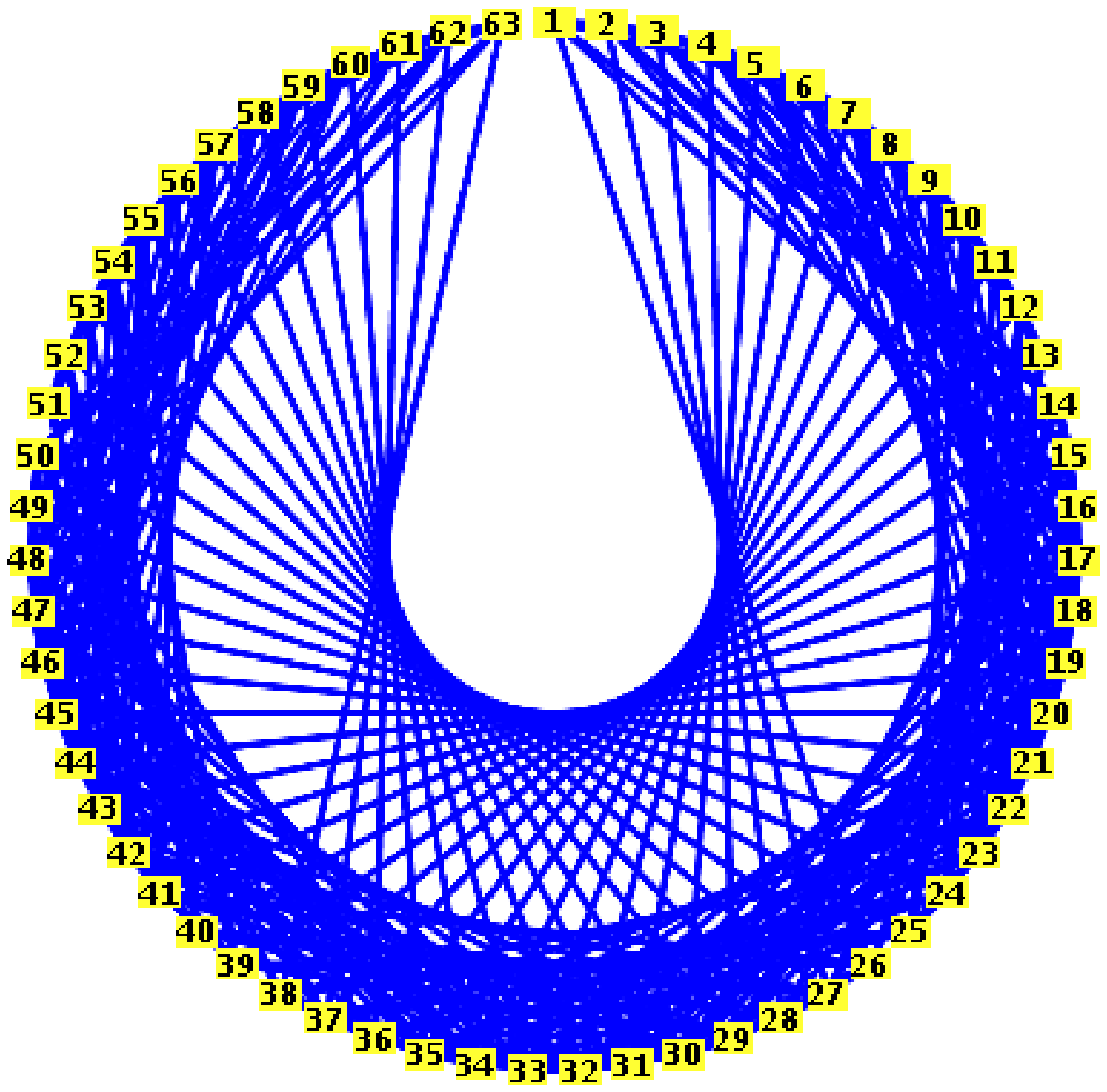}
  \end{minipage}
  \caption{The graph representations obtained from the roots and weights lattices for the irreps of $\spfour$ and $\spsix$, respectively, labeled by $(4,4)$ and $(3,3,3)$. The edges interconnecting the vertices generate a drop shaped internal structure that becomes thinner as the rank of the algebra grows. Furthermore, the $\spsix$ graph has a first layer that is denser and similar to the $\spfour$ graph. That interconnect pattern is due to the $\spfour$ being a subalgebra of the $\spsix$.}
\label{graph}
\end{figure}

\section{Construction of symplectic graphs for network topologies}

Here on we shall refer to the roots and weights lattices of an irrep of a Lie algebra simply as {\em rwl}. A graph shall be denoted by $\grph$ with $\nub$ and $\E$ indicating, respectively, the set of vertices of the graph and its adjacency matrix.

\subsection{Graph's vertices and lattices' nodes correspondence}

The use of lattices for the representation of a graph is helpful for the computation of its adjacency and distance matrices. We start establishing a correspondence rule between the nodes of the lattice and the vertices of the graph. Hence, we shall consider graphs that can be represented as lattices -- $\ltc{\mu}{n}$ -- on a $n$-dimensional space with $\ltc{\mu}{n} = \muzn$ and $\mu$ a given positive integer. A node of $\ltc{\mu}{n}$ is addressed by a $n$-tuple, $\lkpa$, with $\lkpa=(l_1,\dots,l_n)$ and $l_i=0,1, \dots, \mu-1$ for all $i$. The index $\kappa$ is an integer aimed to state the correspondence between a node of the lattice and a vertex of the graph, with $\kappa \in \grph$ obeying
\begin{equation} \label{labeling}
\kappa_\lambda = \sum_{i=1}^n {l_i \, \mu^{i-1}}.
\end{equation}
for a given node $\lkpa$ of $\ltc{\mu}{n}$.

The inverse transformation for the labeling rule of the Eq. (\ref{labeling}) is written in terms of a recursion relation. For a given integer $\klam$ labeling a vertex of a graph, one obtains the address of its corresponding node on a $n$-dimensional lattice as follows:
\begin{equation} \label{inverse-label}
r_n = \klam, \ \ l_i = \lfloor r_i/\mu^{i-1} \rfloor, \ \ r_{i-1} = r_i - l_i\,\mu^{i-1}, \ \ i=n, \dots, 2, \ \ l_1=r_1,
\end{equation}
with $\lfloor . \rfloor$ indicating the floor function.

Here we are going to work with the following lattices on a $n$-dimensional Cartesian space: mesh -- $\ltct{m}\equiv\ltc{\mu}{n}$; hypercubic -- $\ltct{h}\equiv \ltc{2}{n}$; symplectic -- $\ltct{s}$.

\subsection{The hypercubic, mesh and symplectic lattices.}

The $n$-dimensional mesh lattice is defined as $\ltct{m} = \muzn$ and a node of this network is addressed by a $n$-tuple, $\lkpa$, with $\lkpa=(l_1,\dots,l_n)$ and coordinates $l_i=0,\dots,\mu-1$.

The definition of the $n$-dimensional hypercubic lattice is obtained by imposing $\mu=2$ onto the definition of the mesh lattices. Hence, $\ltct{h} = [0,2)^n \cap \mathbb{Z}^n$ and the address of a node is indicated by a $n$-tuple, $\lkpa$, with $\lkpa=(l_1,\dots,l_n)$ and coordinates $l_i=0,1$.

The symplectic lattice, $\ltct{s}$, is defined in terms of its \rwl for an anti-symmetric irrep of the $\spn$ labeled by $(M,\dots,M)$. A weight has coordinates denoted by $w_p$, with $w_p=(m_1,\dots,m_n)$ and the index $p$ defining a label for the weight such that
$$
  p\in\biggl\{ 0,\dots,\frac{(2M+1)^n-1}{2} \biggr\}.
$$
The symplectic lattice can be mapped as a sub-lattice of $\ltc{\mu}{n}$ whose nodes are addressed by $\lkpa=(l_1,\dots,l_n)$ if we impose
\begin{equation}\label{mumrel}
  \mu=2M+1.
\end{equation}
and redefine $m_i$ on the Eq. (\ref{mvalues}) such that
\begin{equation}\label{mMrel}
  m_i \rightarrow m_i+M.
\end{equation}
A weight $w_p$ shall correspond to a node $\lkpa$ if $\kappa=2p$, for all $p\in\{0,\dots,\left[(2M+1)^n-1\right]/2 \}$. Therefore, the nodes of the symplectic lattice correspond to the even labeled nodes of a lattice $\ltc{\mu}{n}$. That occurs due to the adjacency rule between the weights of an irrep of the $\spn$ as we shall demonstrate on the Lemma \ref{l:evendiff}.

\subsubsection{The node adjacency in a symplectic lattice.}

\begin{lemma}\label{l:evendiff}
Let us consider a lattice $\ltc{\mu}{n}$ with nodes addressed by $\lkpa$, with $\mu=2M+1$, and the set of negative root vectors of the $\spn$ algebra. The vector addition of a $n$-tuple $\lkpa$ to a negative root vector of $\spn$ has two possible results: {\em i)} another $n$-tuple labeled by $\lkpal \in \ltc{\mu}{n}$; {\em ii)} a $n$-tuple not belonging to $\ltc{\mu}{n}$. In case of {\em i)} the difference $\Delta \kappa = \kappa - \kappa'$ is an even number. 
\end{lemma}
\begin{proof}[Proof of the Lemma \ref{l:evendiff}.] We are only interested on the condition {\em i)}. We start with a node $\lkpa = (l_1,\dots,l_n)$ and consider its vector addition with a long negative root vector of the form $-2e_i$, namely, $\lkpal = \lkpa+(-2e_i) = (l_1,\dots,l_i-2,\dots,l_n)$. Accordingly with Eq. (\ref{labeling}) the difference 
$$
\Delta \kappa = \kappa - \kappa' = 2 \mu^{i-1},
$$
is an even number. Next we consider the vector addition of the weight and a short negative root vector of the form $-e_i \pm e_j$. That results into $\lkpal = \lkpa+(-e_i \pm e_j) = (l_1,\dots,l_i-1,\dots,l_j\pm1,\dots, l_n)$, such that,
$$
\Delta \kappa = \kappa - \kappa' = \mu^{i-1} \mp \mu^{j-1},
$$
is an even number since $\mu^{i-1}$ and $\mu^{j-1}$ are both odd numbers, see Eq. (\ref{mumrel}).
\end{proof}

\begin{theorem}\label{t:smpltc}
Let us consider a lattice $\ltc{\mu}{n}$, with nodes addressed by $\lkpa$ as in Eq. (\ref{labeling}), and an irrep of $\spn$ given by the maximum weight $(M,\dots,M)$. $\ltct{s}$ is a sub-lattice of $\ltc{\mu}{n}$ if one sets $\mu=2M+1$, and maps each weight vector $w_p$ of the symplectic lattice onto a node of the mesh lattice addressed by $\lkpa$ such that $\kappa = 2p$, with $p=0,1,\dots,(\mu^n-1)/2$. 
\end{theorem}
\begin{proof}[Proof of the theorem \ref{t:smpltc}.] We start noticing that $\mu^n - 1 = (2M+1)^n - 1$ is an even number and that an irrep given by the highest weight $(M,\dots,M)$ has an associated lowest weight $(-M,\dots,-M)$. The mapping of the weight vectors onto the nodes of $\ltc{\mu}{n}$ is done by defining the coordinates of a node $\lkpa$ as $l_i=m_i+M$. Hence, the highest weight vector of the irrep is mapped onto $\lambda_{\mu^n-1}=(2M,\dots,2M)$ and $\lambda_0=(0,\dots,0)$ corresponds to the lowest weight vector. The remaining weight vectors correspondence is obtained by the repeated addition of the negative root vectors of $\spn$ to the node $\lambda_{\mu^n-1}$. The resulting nodes are addressed by $\lkpal$ and the difference $\mu^n-1 - \kappa'$ is an even number, as stated on Lemma \ref{l:evendiff}.
\end{proof}

A corollary of the Theorem \ref{t:smpltc} is that the nodes of $\ltct{s}$ mapped onto $\ltc{\mu}{n}$ shall be addressed by two types of $n$-tuples. We designate as {\em bosonic} the nodes whose $\lkpa$ are of the form $(2q_1,\dots,2q_n)$ and as {\em fermionic} the nodes that have coordinates of the form $2q_i+1$ and of the form $2q_j$. A fermionic node has an even number of coordinates of the form $2q_i+1$. Note that $q_i \in \{0, \dots, M\}$ for the coordinates of the form $2q_i$ and that for coordinates of the form $2q_i+1$ we have $q_i \in \{0, \dots, M-1\}$.

\subsection{Adjacency properties on lattices and graphs}

The linking between two first neighbor nodes of a lattice can be expressed in terms of the vector connecting them. The set of all vectors connecting the first neighbor nodes of a node of the lattice is the set of all edges of the corresponding graph. Hence, the correspondence between the lattice's vectors and the graph's edges can be described in terms of the unit vectors $e_1,\dots,e_n$ spanning the $n$-dimensional space where the lattice is defined. In the {\em hypercubic} and the {\em mesh} topologies the vectors linking a node to its neighbor is given in terms of the unit vectors $\pm e_i$. The {\em symplectic} topology has its set of edges connecting a node to its first neighbors given by the root vectors, $\pm 2e_i$ and $\pm e_i \pm e_j$, as defined at the Eq. (\ref{roots1}).

The final step is to determine the adjacency matrix of the graph in terms of the edges of a lattice. The set of all first neighbors of a node $\lkpa$ of a lattice is determined by summing $\lkpa$ to the vectors indicating the edges of the topology.

In a mesh topology, the nodes $\lkpal$ that are first neighbors of a node $\lkpa$ satisfy:
\begin{equation}\label{msh-nbr}
\kappa' = \kappa \pm \mu^{i-1},  \ \ i=1,\dots,n, \ \ {\rm such \ \ that} \ \ 0\le \kappa' \le \mu^n-1,
\end{equation}
ensures $\lkpal\in\ltct{m}$. That is obtained with the Eq. (\ref{labeling}).

Similarly, in a hypercubic topology the node $\lkpa$ has as first neighbors the nodes $\lkpal$ which labels obey:
\begin{equation}\label{hpc-nbr}
\kappa' = \kappa \pm 2^{i-1}, \ \ i=1,\dots,n, \ \ {\rm such \ \ that} \ \ 0\le \kappa' \le 2^n-1,
\end{equation}
ensures $\lkpal\in\ltct{h}$ and again we have used the Eq. (\ref{labeling}).

\begin{figure}[!h]
  \centering
  \begin{minipage}[b]{0.32\linewidth}
    \includegraphics[width=\linewidth]{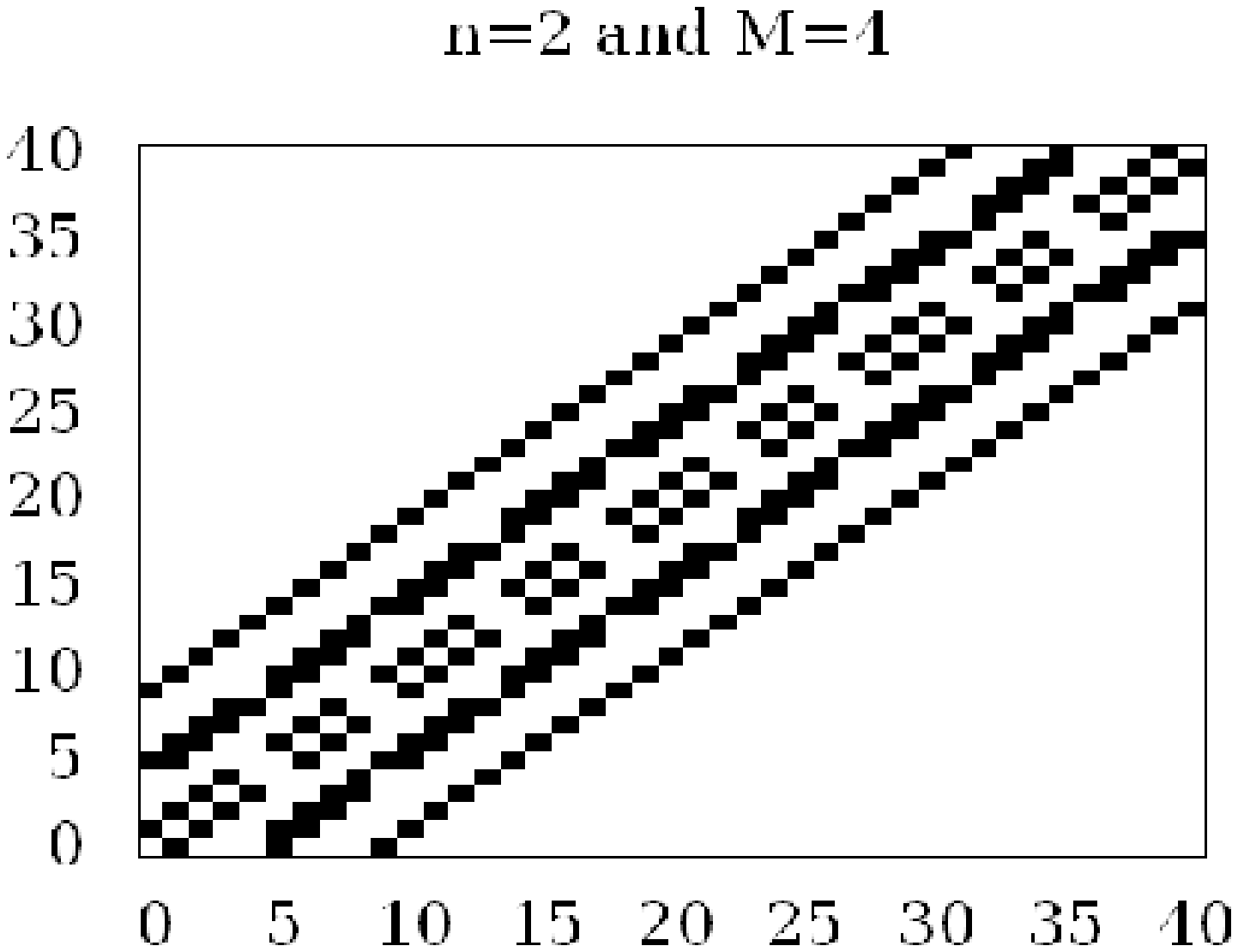}
    \end{minipage}
  \hfill
  \begin{minipage}[b]{0.32\linewidth}
    \includegraphics[width=\linewidth]{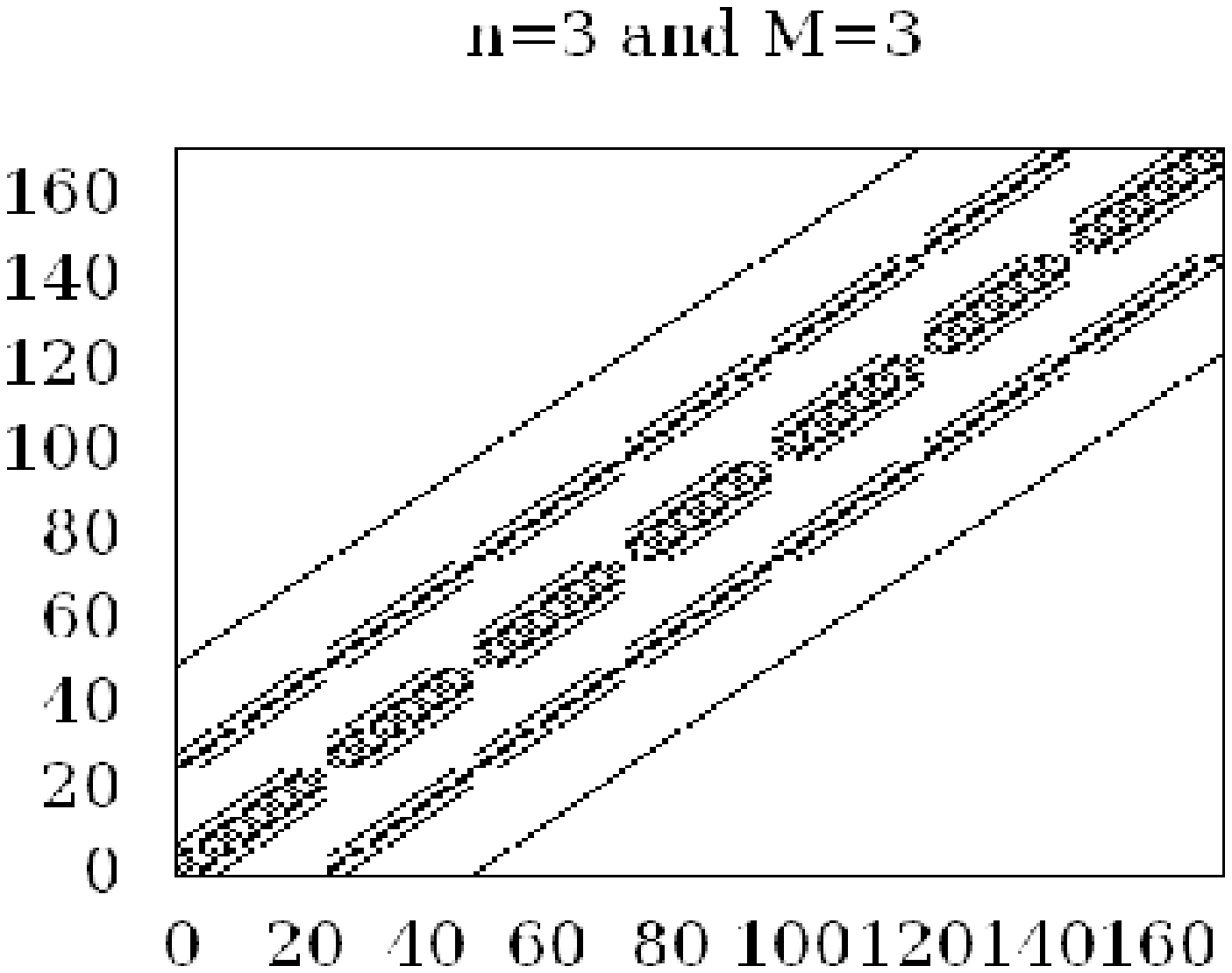}
  \end{minipage}
  \hfill
  \begin{minipage}[b]{0.32\linewidth}
    \includegraphics[width=\linewidth]{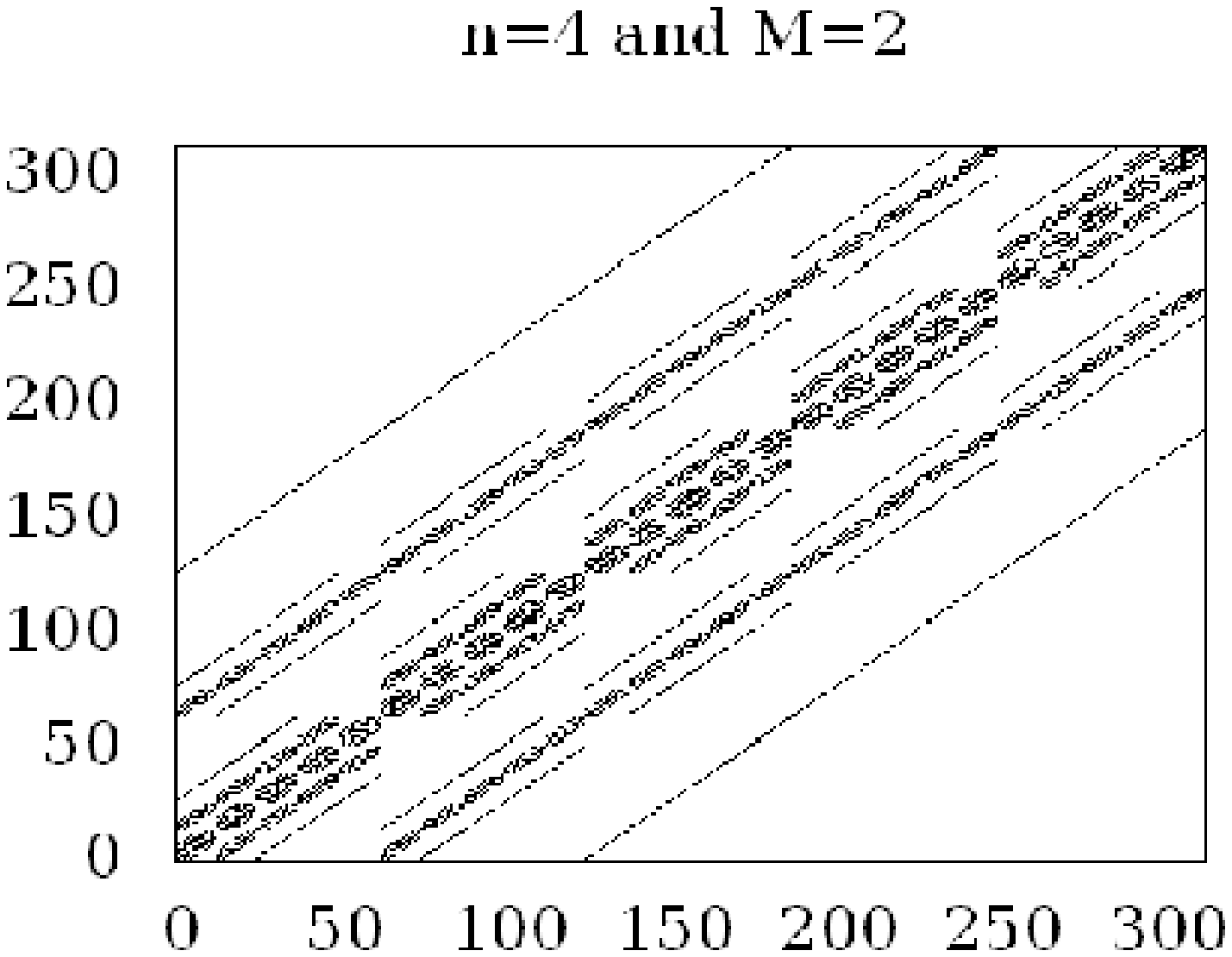}
  \end{minipage}
  \caption{{\bf Adjacency matrix.} The pattern of the adjacency matrix of the symplectic graphs obtained from the roots and weights lattices for the irreps of $\spfour$, $\spsix$, and $\speight$, respectively, labeled by $(4,4)$, $(3,3,3)$, and $(2,2,2,2)$. The \rwl have, respectively, 41, 172, and 313 vertices. The most internal pattern of the non-principal diagonal of the adjacency matrix of a $\spnone$ recovers the patterns of the $\spn$ \rwl adjacency matrix. This is due to $\spnone \supset \spn$ and the remaining elements are indicating the occurrence of new interconnects of the greater algebraic structure.}
\label{adjm}
\end{figure}

In the symplectic topology we consider a node $w_p$ of $\ltct{s}$ associated to the node $\lkpa$ of $\ltc{\mu}{n}$. The first neighbors $w_{p'}$ of $\ltct{s}$ associated to the nodes $\lkpal$ of $\ltc{\mu}{n}$ are such that:
\begin{eqnarray}
\kappa' = 
\left\{
\begin{array}{ll}
\kappa \pm 2(2M+1)^{i-1}, & i=1,\dots,n, \\
\kappa \pm (2M+1)^{i-1} \pm (2M+1)^{j-1}, & 1 \le i < j \le n,
\end{array}
\right. \label{stc-nbr} 
\end{eqnarray}
is obtained with the Eq. (\ref{labeling}). To ensure that $\lkpal \in \ltct{s}$ one imposes that $0\le \kappa' \le (2M+1)^n-1$.

The Eqs. (\ref{msh-nbr}), (\ref{hpc-nbr}), and (\ref{stc-nbr}) determine the entries labeled by $\kappa$ and $\kappa'$ of the adjacency matrix of the mesh, hypercubic and symplectic topologies that shall have values equal to 0 or 1. The Figure \ref{adjm} shows three examples of adjacency matrices for the symplectic graphs. 

\subsection{The node distance in a lattice}

The next two theorems are dedicated to the calculation of the distance between two nodes of a lattice by considering its nodes coordinates. One may use those theorems to construct the distance matrix of a lattice and, by means of the correspondence rule of the Eqs. (\ref{labeling}) and (\ref{inverse-label}), interpret it as a distance matrix of the associated graph. We first show the theorem for the mesh and hypercubic lattices and then for the symplectic topology. 

\begin{theorem}\label{t:dstnc}
Let us consider the mesh and the hypercubic lattices as defined previously and denote the distance between two nodes $\lkpa=(l_1,\dots,l_n)$ and $\lkpal=(l_1',\dots,l_n')$ by, respectively, $\dist{m}$ and $\dist{h}$. The distance between the nodes $\lkpa,\lkpal$ of the mesh and hypercubic lattices is given by:
\begin{equation} \label{distance}
\dist{m,h} = \sum_{i=1}^n | l_i - l_i' |.
\end{equation}
\end{theorem}
\begin{proof}[Proof of the theorem \ref{t:dstnc}] The vector connecting the nodes $\lkpa$ and $\lkpa'$ is $\Delta \lkpa = (\Delta l_1, \dots, \Delta l_n)$, and in terms of the unit vectors $e_i$ one may write $\Delta \lkpa = \Delta l_1 e_1 + \dots +  \Delta l_n e_n $. Since the addresses of the nodes $\lkpa$ and $\lkpa'$ are built only with integers, the elements $\Delta l_i$ are also integers. The distance between the nodes is the number of unit vectors needed to represent the vector $\Delta \lkpa$ which is $\sum_{i=1}^n |\Delta l_i|$.
\end{proof}

\begin{figure}[!h]
  \centering
  \begin{minipage}[b]{0.32\linewidth}
    \includegraphics[width=\linewidth]{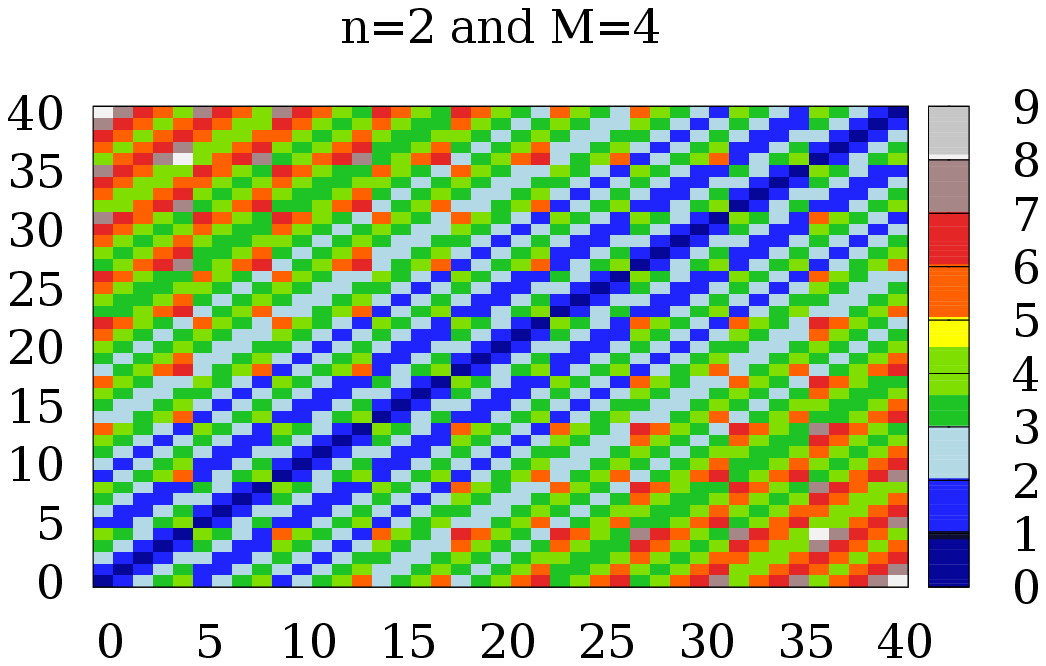}
    \end{minipage}
  \hfill
  \begin{minipage}[b]{0.32\linewidth}
    \includegraphics[width=\linewidth]{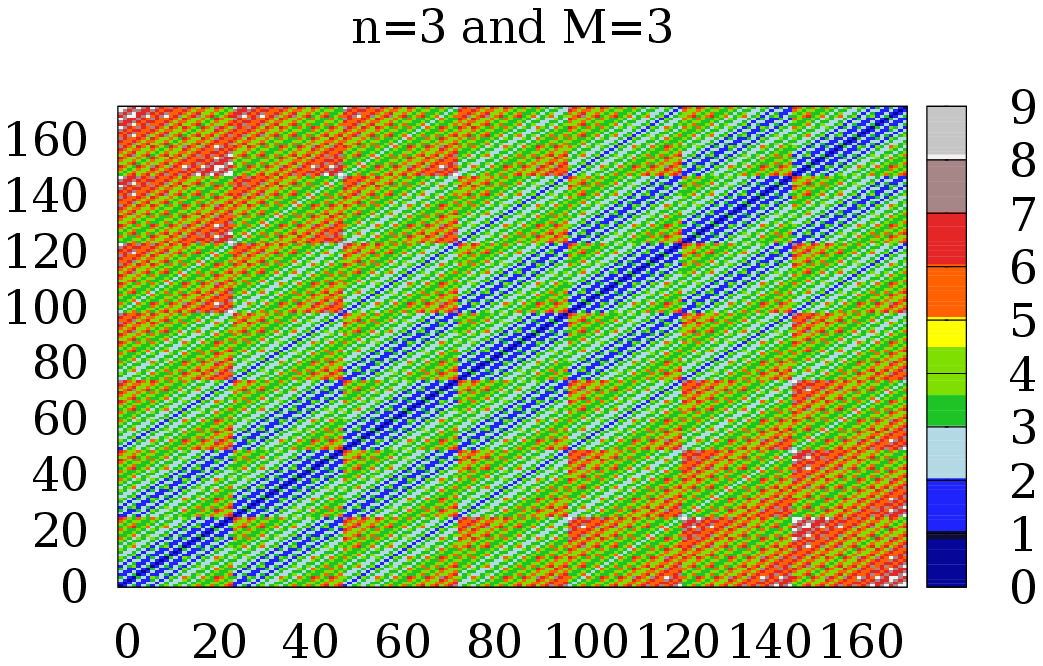}
  \end{minipage}
  \hfill
  \begin{minipage}[b]{0.32\linewidth}
    \includegraphics[width=\linewidth]{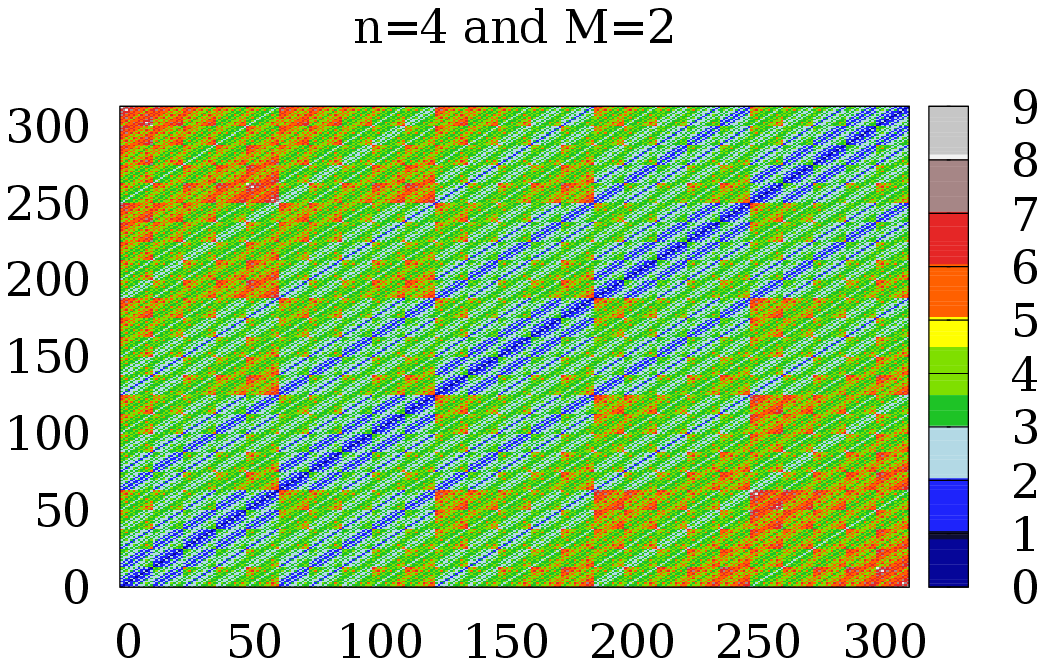}
  \end{minipage}
  \caption{{\bf Distance matrix.} The pattern of the distance matrix of the symplectic graphs obtained from the roots and weights lattices for the irreps of $\spfour$, $\spsix$, and $\speight$, respectively, labeled by $(4,4)$, $(3,3,3)$, and $(2,2,2,2)$.}
\label{dstncm}
\end{figure}

\begin{theorem}\label{t:dstnc1}
Let us consider the symplectic lattice and denote by $\dist{s}$ the distance between two nodes $\lkpa=(l_1,\dots,l_n)$ and $\lkpa'=(l_1',\dots,l_n')$, with $\lkpa,\lkpal \in \ltct{s}$. The distance between the two nodes is:
\begin{equation} \label{distance1}
\dist{s} = \sum_{i=1}^n \frac{| l_i - l_i' |}{2}.
\end{equation}
\end{theorem}

\begin{proof}[Proof of the theorem \ref{t:dstnc1}] The proof of the Eq. (\ref{distance1}) is presented considering distances between the following types of nodes: bosonic-bosonic; bosonic-fermionic; fermionic-fermionic. The distance between two bosonic nodes is demonstrated following the same scheme as the proof of the Theorem (\ref{t:dstnc}). The division by two is due to the fact the length of a single edge of the symplectic lattice is equal to two as given by Eq. (\ref{roots1}).

The distance between bosonic and fermionic nodes addressed by, respectively, $\lkpa$ and $\lkpal$, is given in terms of the vector $\Delta \lkpa = (\Delta l_1, \dots, \Delta l_n)$, with $\Delta l_i = |l_i-l_i'|$. We consider an integer $p$ and notice that the fermionic nodes have $2p$ components with coordinates given by odd numbers and the remaining $n-2p$ coordinates given by even numbers. For simplicity, we deal with the case of only two coordinates of the fermionic node given by odd numbers. In that case, $\Delta \lkpa = (\Delta l_1,\dots,\Delta l_i, \dots,\Delta l_j, \dots, \Delta l_n)$, with $\Delta l_i=2p_i+1$ and $\Delta l_j=2p_j+1$. The remaining elements $\Delta l_1, \Delta l_2, \dots, \Delta l_n$ are given by $2p_1, 2p_2, \dots, 2p_n$, respectively, and note that all $p$'s belong to $\left\{0,1,2,\dots\right\}$. The vector $\Delta \lkpa$ can be written as a linear combination of the form $p_1(2e_1)+\dots+(e_i+e_j) + p_i(2e_i) + p_j(2e_j)+\dots+p_n(2e_n)$. Note that the vectors $2e_1,\dots, 2e_n$ and $(e_i+e_j)$ all belong to the set of roots of the symplectic algebra, as defined at the Eq. (\ref{roots1}). Therefore, one combines $1+\sum_{i=1}^n p_i = \sum_{i=1}^{n}\Delta l_i / 2$ root vectors to connect a bosonic node to a fermionic one. It is straightforward to generalize this form for the case of $2p$ coordinates of $\Delta \lkpa$ given by odd numbers. That will result into a combination of $p+\sum_{i=1}^n p_i = \sum_{i=1}^{n}\Delta l_i / 2$ root vectors to connect the bosonic and the fermionic nodes.

The distances between fermionic nodes, addressed by $\lkpa$ and $\lkpal$ is also given in terms of the vector $\Delta \lkpa = (\Delta l_1, \dots, \Delta l_n)$, where the vector components, $\Delta l_i = |l_i-l_i'|$, can be written in two different forms: $2p_i$, with $p_i \in \{0,\dots,M\}$ or $2p_i+1$, with $p_i \in \{ 0,\dots, M-1\}$. Given a positive integer $q$, one has $2q$ coordinates of $\Delta \lkpa$ written as $2p_i+1$. Hence, similarly to the previous case, one combines $q+\sum_{i=1}^n p_i = \sum_{i=1}^{n}\Delta l_i / 2$ root vectors of the symplectic algebra to connect two fermionic nodes.
\end{proof}

The Figure \ref{dstncm} shows the distance matrices for the \rwl corresponding to the adjacency matrices presented at the Figure \ref{adjm}. The symmetry of the \rwl is reflected on the repeating patterns of the distance matrices. 

\subsection{Lattice density}

For a comparative analysis between different topologies it is also useful to define the network density, denoted by $\rho$. Let us consider a graph $\grph$ composed by $\nu$ nodes and with a diameter $L$. Then, if one may map that graph into a $n$-dimensional lattice it is possible to define a "`graph volume"' as $L^n$. Then we define a quantity named graph density by:
\begin{equation}
\rho = \frac{\nu}{L^n}.
\end{equation}
$\rho$ gives the amount of vertices packed into the volume of a certain graph topology. This quantity is useful to compare two network topologies having the same diameters and different number of nodes. The higher the density is, the better is the packing of the network topology, as it accommodates a bigger number of nodes.

\subsection{A comparison among the hypercubic, mesh, and symplectic topologies}

\subsubsection{Computation of the distance matrices}

We perform a comparative analysis of the three network topologies presented in the previous sections. That is useful to show the main characteristics of the symplectic topology in comparison with well known topologies as hypercubic and mesh. We shall denote the hypercubic graph by $\hgrp$, where $\nub$ is given by $\{0, \dots, 2^n - 1 \}$ 
and $n$ indicates the dimension of the hypercube. The mesh topology is denoted by $\mgrp$, where $\nub$ is given by $\{0, \dots, \mu^n - 1 \}$ and $n$ indicates 
the dimension of the mesh. Finally, we denote the symplectic graph by $\sgrp$, with nodes labeled by an integer $\{0,\dots,\frac{(2M+1)^n-1}{2}\}$, $n$ defining the rank, and $M$ defining the maximum weight of the anti-symmetric irrep of the symplectic algebra. The adjacency matrix for the graphs $\hgrp$ and $\mgrp$ can be constructed with the help of the Eq. (\ref{inverse-label}) and the Theorem (\ref{t:dstnc}). For the case of the $\sgrp$ one first multiplies the labels of the nodes by two and then applies the Eq. (\ref{inverse-label}) and the Theorem (\ref{t:dstnc1}).

\subsubsection{Graph diameters and vertex connectivity}

Since these two properties are well-known for the mesh and the hypercube, we just present our results for the symplectic topology. The graph diameter of the symplectic topology can be obtained by converting the labels $0$ and $2 \times \frac{(2M+1)^n-1}{2}$ into addresses of the lattice, by using the Eq. (\ref{inverse-label}). Then one uses the Theorem (\ref{t:dstnc1}) for calculating the symplectic lattice diameter as $L_s = M\, n$ as it can be verified by inspection of the Fig. \ref{smpltc}. To evaluate the connectivity, we consider the $n$-dimensional hypercube enveloping the anti-symmetric symplectic lattice. The nodes at the vertices of this hypercube have connectivity given by $n(n+1)/2$ which is the result of summing the all strictly positive root vectors of a symplectic algebra. The maximal connectivity of the symplectic algebra is given by $\epsilon = 2n^2$, which is the total amount of positive and negative roots of a symplectic algebra.

We may resume the properties of the hypercubic, mesh and symplectic topologies in a table. The hypercubic topology is indicated by $H(2,n)$, with $n$ denoting the dimension of the hypercube. The mesh lattice with length $\mu$ and dimension $n$ is denoted by $H(\mu, n)$. The anti-symmetric weight lattice of the irrep of a rank $n$ symplectic algebra, labeled by an $n$-tuple $(M,\dots,M)$ is indicated by $S(M,n)$.

\begin{table}[!ht]
\centering
\begin{tabular}{|c|c|c|c|}
\hline
 & $H(2,n)$ & $H(\mu,n)$ & $S(M,n)$ \\
\hline
\hline
$\nu$ & $2^n$ & $\mu^n$ & $\displaystyle \frac{(2M+1)^n+1}{2}$ \\
\hline
$L$ & $n$ & $(\mu-1)\,n$ & $M n$ \\
\hline
$\overline{\epsilon}$ & $n$ & $2n$ & $2n^2$ \\
\hline
$\underline{\epsilon}$ & $n$ & $n$ & $\frac{n(n+1)}{2}$ \\
\hline
$\rho$ & $\displaystyle \left(\frac{2}{n} \right)^n$ & $\displaystyle \left(\frac{\mu}{(\mu-1) \, n} \right)^n$ & $\displaystyle \frac{1}{2} \displaystyle \frac{(2M+1)^n + 1}{(Mn)^n}$  \\ [2ex]
\hline
\end{tabular}
\vspace*{10pt}
\caption{A table resuming the main characteristics of the topologies evaluated in this manuscript. $\nu$ indicates the total number of nodes of a network; $L$ indicates the network diameter; $\overline{\epsilon}$ ($\underline{\epsilon}$) gives the maximum (minimum) connectivity of a node of the network; $\rho$ gives the network density.}
\label{tablechar}
\end{table}

\subsubsection{Density for the mesh and hypercubic topologies}

Let us consider the variables of the Tab. \ref{tablechar} for each topology. We start with a comparison between the lattices presenting the same amount of nodes. $H(\mu,n')$ and $H(2,n)$ have the same number of nodes for $\mu = 2^m$ and $n'=n/m$. The ratio between the densities of the hypercubic ($\rho_h$) and the mesh ($\rho_m$) topologies results $\rho_h / \rho_m = \left[ 2(1-1/\mu) \right]^{n}$. The proportionality factor tends to infinite for two conditions, a finite $\mu$ and $n \rightarrow \infty$ or for the opposite, that is $\mu \rightarrow \infty$ and a finite $n$. Hence, the hypercubic topology provides a higher density than the mesh topology and can be used as a reference interconnect topology.

\begin{figure}[!ht]
\centering
\includegraphics[width=0.45\linewidth]{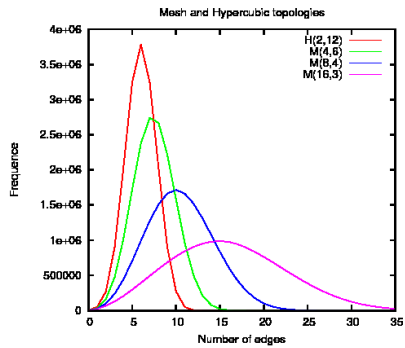}
\caption{The path length distribution for the hypercubic and mesh topologies for $\nu=4096$.}
\label{meshyp}
\end{figure}

The Figure \ref{meshyp} shows the comparison between the hypercubic and the mesh topologies for $4096$ nodes. The hypercubic topology shows a higher packing in comparison with the mesh topology. Notice that as the number of dimensions of the mesh lattice decreases and $\mu$ increases, the spreader the distribution of distances becomes.

\subsubsection{Hypercubic and symplectic topologies of the same diameter -- I} The number of nodes on the symplectic and the hypercubic topologies cannot be the same. Hence we compare lattices $H(2,n)$ and $S(M,n')$ such that their diameters are the same. Hence, we have the condition $n=M\,n'$. We fix $M$ to be an integer greater than one, such that $n'=\frac{n}{M}$.

\begin{figure}[!h]
\centering
\includegraphics[width=0.45\linewidth]{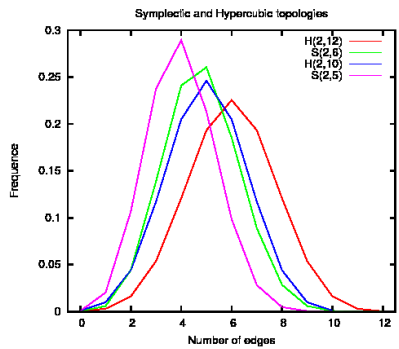}
\caption{The hypercubic and the symplectic topologies have the same network diameters, $L=10$ and $L=12$. The number of nodes for the hypercubic topologies are 1024 and 4096, respectively, at dimensions 10 and 12. The corresponding symplectic lattices are at the dimensions 5 and 6, with the number of nodes given by 1563 and 7813.}
\label{hypsym}
\end{figure}

The Figure \ref{hypsym} shows the result of a comparison between the $H(2,n)$ and $S(M,n/M)$ topologies. The path length distributions of the two graphs will present the same diameter $n$. The symplectic topology have left displaced path length distributions in comparison with a hypercubic topology of the same diameter. On the other hand, the symplectic topology accommodates a bigger amount nodes and, hence, has a higher density than the hypercubic arrangement.

\begin{figure}
  \centering
  \includegraphics[width=0.45\linewidth]{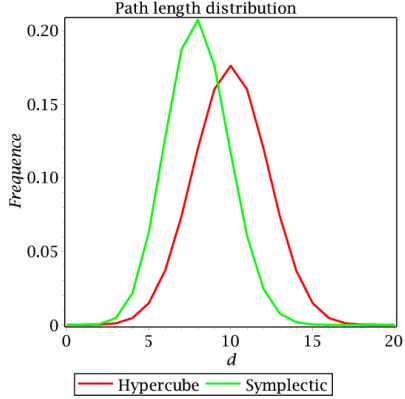} \vspace{1cm}
  
      \caption{Comparison between the symplectic and the hypercubic topologies. We consider the hypercubic graph for $n=20$ which results on 1,048,576 vertices. The graph diameter is 20 while the average path distance is 11 edges. The symplectic graph has parameters $n=10$ and $M=2$ which results on 4,882,813 vertices. The graph diameter is also 20 with average path distance given by 8.}
\label{f:n20}
\end{figure}

The Figure \ref{f:n20} shows a comparison on the path length distribution between the hypercubic and symplectic topologies. We are considering the 20 dimensional hypercube with $\sim$ 1M vertices. We fixed a symplectic topology having a diameter 20, the same as the hypercube, and constructed a graph having $\sim$ 4.8M vertices. Again, the path length distribution on the symplectic topology is displaced to the left in comparison with the hypercubic. Therefore, the average path length on the symplectic topology is smaller than on the hypercubic. Furthermore, the shape of the distributions of path length on the symplectic topologies is preserved and it permits us to extrapolate our results for \rwl having a greater number of nodes.

\subsubsection{Hypercubic and symplectic topologies of the same diameter -- II}

The Figure \ref{hypsym1} shows a comparison between the symplectic and the hypercubic topologies for the same value of $n$ and diameter $L$. Hence, we have $M=1$ and the path length distributions of the symplectic topology shall be thinner than that for the hypercube. Again, the symplectic topology shall provide a greater number of nodes and a higher density. Indeed, let us compare the symplectic ($\rho_s$) and the hypercubic ($\rho_h$) densities in terms of their ratio for the same dimension $n$. That results $$\frac{\rho_s}{\rho_h} = \frac{1}{2}\frac{(2M+1)^n+1}{(2M)^n}.$$ For $n>1$ and $M$ finite the ratio between the two densities is greater than one. As $M$ goes to infinity, the density of the hypercubic topology becomes greater than the symplectic by a factor two. That is because the symplectic lattice's diameter shall also go to infinity. It is striking, though, that the hypercubic topology shall be, maximally, only as twice as more ``packed'' than the symplectic.

\begin{figure}[!h]
\centering
\includegraphics[width=0.45\linewidth]{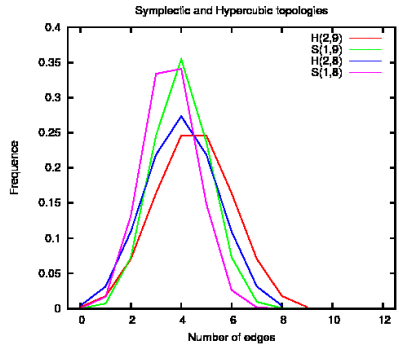}
\caption{The hypercubic and the symplectic topologies have the same diameters, $L=8$ and $L=9$. The number of nodes for the hypercubic topologies are 256 and 512, respectively, in dimensions 8 and 9. The corresponding symplectic lattices are in the same dimensions, 8 and 9, have 3281 and 9842 nodes, respectively.}
\label{hypsym1}
\end{figure}

Furthermore, one may notice that for the case of the Figure \ref{hypsym1} the ratio between the two densities results into $\frac{3^n+1}{2^{n+1}}$ and this number grows with $n$. Therefore we show that, for the same network diameter, the symplectic topology provides us with a systematic approach to combine a much bigger amount of nodes in a network.

\subsubsection{Comparing the symplectic and the mesh topologies}

The Figure \ref{symesh} shows a comparison between the symplectic and the mesh topologies when they have the same dimension $n$ and diameter $L$. As it can be viewed on the Table \ref{tablechar} that imposes the topologies $S(M,n)$ and $H(\mu,n)$ to obey $M=\mu-1$. The Figure \ref{symesh} shows the path length distributions for the two topologies. Inspection is enough to notice that the symplectic topology generates a distance distribution that is slightly thinner and displaced to the left in comparison with the mesh.

\begin{figure}
\centering
\includegraphics[width=0.45\linewidth]{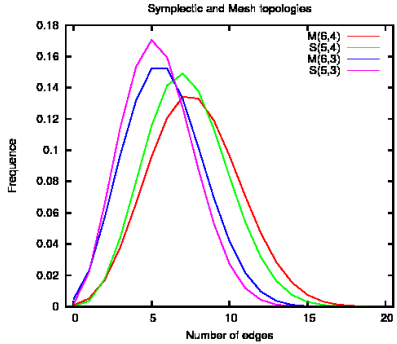}
\caption{The mesh and the symplectic topologies have the same network diameters, $L=20$ and $L=15$. The mesh topology have 216 and 1296 nodes while the corresponding symplectic topology has 666 and 7321, both lattices respectively at dimensions 3 and 4.}
\label{symesh}
\end{figure}

Again, the symplectic topology provides a strategy to accommodate a greater amount of nodes than the mesh in the same lattice volume. That is verified by evaluating the ratio between the mesh ($\rho_m$) and symplectic ($\rho_s$) lattice densities. That results into $$\frac{\rho_s}{\rho_m} = 2^{n-1} \left[ \left( 1-\frac{1}{2\mu}\right)^n+1\right]$$. For $\mu$ going to infinity the ratio between the two densities tends to be $2^n$. For a fixed $\mu$ and $n>>1$ that ratio grows as $2^{n-1}$. Either way the symplectic density is higher than that of the mesh.

\section{Conclusions}

In this manuscript we have shown that group theoretical tools may play a pivotal role on the design of network interconnects for supercomputers. That was shown in terms of the roots and weights lattices of the irreps of the symplectic algebras, one of the classical symmetries in the Cartan classification. A strategy to establish the correspondence between those lattices and the graphs was established. The use of roots and weights lattices turns the calculation of the distances between its nodes into a task based on the "`taxicab geometry"'. Therefore, it simplifies the calculation of the vertices' distances on the corresponding graph. The usefulness of such a technique can be appreciated on a more intricate graph, as it is the case of the here introduced symplectic topology. The analysis of the symplectic topology has demanded simple geometrical concepts instead of eventually cumbersome discrete mathematical methods. 

To understand the applicability of the symplectic topologies we have compared it with the well known mesh and hypercubic topologies on Table \ref{tablechar}. As a first step, we have compared the main characteristics of those two families of topologies, such as the number of vertices, the diameter, and maximal (and minimal) vertex degree. We have verified that for a given diameter, the symplectic topology provides graphs which number of vertices shall be greater than that of the mesh and the hypercube. 

A coarse-grained understanding of the differences between the symplectic, mesh and hypercubic topologies can be obtained by considering the concept of lattice density, as we have introduced. That is the ratio of the number of nodes of the lattice for its volume on the $n$ dimensional Cartesian space. Here, the volume is defined as $L^n$ and the lattice density gives the nodes concentration. A high lattice density implies on smaller node to node distance measure in terms of the number of edges connecting them. We also warn the reader that this measure do not takes into account more delicate information about the lattice structure. The node to node path length in a lattice obeys a distribution and a more precise analysis of the lattice shall rely on the investigation of those distributions.

Indeed, we have carried out such an analysis and show it on the Figures \ref{meshyp}, \ref{hypsym}, \ref{f:n20}, \ref{hypsym1}, \ref{symesh}. We have structured our analysis considering two possibilities: graphs having the same number of vertices and graphs having the same diameter. The Figure \ref{meshyp} shows the comparison between distances distributions considering the mesh and hypercubic topologies. One observes that the hypercubic topology provides a reference topology for a graph having small diameter and big number of vertices. Another analysis is introduced by considering different topologies having the same diameter. That was shown on the Figures \ref{hypsym}, \ref{f:n20}, \ref{hypsym1}, \ref{symesh}. Exclusive inspection of the path length distributions maybe misleading due to their similarity on shape. That analysis is complemented by the consideration of the number of vertices of the graph or their density. Hence, it is straightforward to conclude that the symplectic topology enables the reduction the typical node-node distance in a graph with a systematic procedure for the construction of its adjacency.

\section*{Acknowledgments}
AFR thanks CAPES, FAPESP and Science Foundation for Youths of Shandong Academy of Sciences of China (N0. 2014QN010)  for financial support. The authors are thankful for Changnian Han for help with parallel code.

\end{document}